\documentclass[a4paper,USenglish,cleveref,autoref,thm-restate,pdfa]{lipics-v2021}
\hideLIPIcs\def\subjclassHeading{}\ccsdesc{}\global\renewcommand\ccsdesc[2][100]{}
\nolinenumbers %uncomment to disable line numbering 
\allowdisplaybreaks[1]

\title{A Note on the Complexity of Directed Clique}
\author{Grzegorz Gutowski}
{Institute of Theoretical Computer Science, Faculty of Mathematics and Computer Science, Jagiellonian University, Krak{\'o}w, Poland}
{grzegorz.gutowski@uj.edu.pl}
{https://orcid.org/0000-0003-3313-1237}
{Partially supported by grant no.~2023/49/B/ST6/01738 from National Science Centre, Poland.}

\author{Mikołaj Rams}
{Institute of Theoretical Computer Science, Faculty of Mathematics and Computer Science, Jagiellonian University, Krak{\'o}w, Poland}
{mikolaj.rams@student.uj.edu.pl}
{https://orcid.org/0009-0004-9142-9167}
{}

\authorrunning{G. Gutowski, M. Rams}
\Copyright{Grzegorz Gutowski, Mikołaj Rams}

\keywords{Directed Clique, Computational Complexity, Polynomial Hierarchy}
\ccsdesc[500]{Theory of computation}

\usepackage{amsmath,amsfonts,amssymb,amsthm}
\usepackage{graphicx}
\usepackage{tcolorbox}
\usepackage{tikz}
\usetikzlibrary{arrows}
\usetikzlibrary{backgrounds}
\usetikzlibrary{calc}
\usetikzlibrary{decorations.pathmorphing}
\usetikzlibrary{decorations.pathreplacing}
\usetikzlibrary{decorations.shapes}
\usetikzlibrary{fit}
\usetikzlibrary{patterns}
\usetikzlibrary{shapes.geometric}
\usepackage[textsize=small]{todonotes}
\usepackage{xcolor}
\usepackage{xspace}

\newcommand{\defproblem}[4]{
  \begin{tcolorbox}%
    \nolinenumbers
    \hspace{0ex}\hspace*{-2.8ex}
    \begin{minipage}{0.99\textwidth}
      \vspace{0ex}\vspace*{-1ex}
      \begin{tabular}{@{}l@{~~}p{0.9\textwidth}@{}}
        {\sf\bfseries\color{gray} Problem:} & #1\\[.1ex]
        {\sf\bfseries\color{gray} Input:} & #2\\[.1ex]
        {\sf\bfseries\color{gray} #4:} & #3\\[-1ex]
      \end{tabular}
    \end{minipage}
  \end{tcolorbox}
}
\newcommand{\defdecproblem}[3]{\defproblem{#1}{#2}{#3}{Question}}

\let\geq\geqslant
\let\setminus\smallsetminus

\let\rho\varrho

\newcommand{\POL}{\ensuremath{\mathsf{P}}\xspace}
\newcommand{\NP}{\ensuremath{\mathsf{NP}}\xspace}
\newcommand{\CONP}{\ensuremath{\mathsf{co}\text{-}\mathsf{NP}}\xspace}

\newcommand{\PHS}[1]{\ensuremath{\Sigma^\mathsf{P}_{#1}}\xspace}

\newcommand{\SPTWO}{\PHS{2}}

\newcommand{\Pclique}{\textsc{Clique}\xspace}
\newcommand{\Pdirclique}{\textsc{DirectedClique}\xspace}
\newcommand{\Ptourclique}{\textsc{TournamentClique}\xspace}
\newcommand{\Ptdirclique}[1][t]{\textsc{$#1$-DirectedClique}\xspace}
\newcommand{\Pttourclique}[1][t]{\textsc{$#1$-TournamentClique}\xspace}
\newcommand{\PSPTWO}{\textsc{Existential-}2\textsc{-Level-}3{-CNF}\xspace}
\newcommand{\PSPTWOSAT}{\textsc{Existential-}2\textsc{-Level-SAT}\xspace}

\newcommand{\abrac}[1]{\left\langle#1\right\rangle}
\newcommand{\set}[1]{\left\{#1\right\}}
\newcommand{\norm}[1]{\left|#1\right|}

\newcommand{\ora}[1]{\overrightarrow{#1}}
\DeclareMathOperator{\diomega}{\ora \omega}

\DeclareMathOperator{\backedge}{BE}

\definecolor{dark blue}{rgb}{0.121,0.47,0.705}
\let\emph\relax\DeclareTextFontCommand{\emph}{\color{dark blue}\em}

\begin{document}

\maketitle

\begin{abstract}
    For a directed graph~$G$, and a linear order $\ll$ on the vertices of $G$,
    we define \emph{backedge graph} $G^\ll$ to be the undirected graph on the same vertex set with edge $\set{u,w}$ in $G^\ll$ if and only if $(u,w)$ is an arc in $G$ and $w \ll u$.
    The \emph{directed clique number} of a directed graph~$G$ is defined as the minimum size of the maximum clique in the backedge graph $G^\ll$ taken over all linear orders $\ll$ on the vertices of~$G$.
    A natural computational problem is to decide for a given directed graph~$G$ and a positive integer~$t$, if the directed clique number of~$G$ is at most~$t$.
    This problem has polynomial algorithm for $t=1$ and is known to be \NP-complete for every fixed $t\ge3$, even for tournaments.
    In this note we prove that this problem is \SPTWO-complete when $t$ is given on the input. 
\end{abstract}

\section{Introduction}
All graphs discussed in this paper are finite, simple, and loopless.
For directed graphs we allow for parallel arcs in opposite directions.
We do not allow for parallel arcs in the same direction.
The vertex set and the edge set of an undirected graph~$G$ are denoted by~$V(G)$ and~$E(G)$.
The vertex set and the arc set of a directed graph~$G$ are denoted by~$V(G)$ and~$A(G)$.
For a subset~$U\subseteq V(G)$, $G[U]$ denotes the subgraph of~$G$ \emph{induced} by~$U$. 
A directed graph $T$ is a \emph{tournament} when for each pair $u,w \in V(G)$ of vertices, exactly one of the $(u,w)$ or $(w,u)$ is an arc in $A(G)$.
For a directed graph~$G$, and a linear order $\ll$ on the vertex set $V(G)$ we define the set of \emph{backedges} to be:
\[
    \backedge(G,\ll) = \set{ \set{u,w} : (u,w) \in A(G) \land w \ll u }\text{.}
\]
We define the \emph{backedge graph} $G^\ll$ to be the undirected graph with vertex set $V(G)$ and edge set $\backedge(G,\ll)$.
We say that a directed graph $G$ is \emph{transitive} if it does not contain a directed cycle.
The \emph{transitive tournament} $TT_n$ with $n$ vertices is the only tournament with $n$ vertices that is transitive.

The \emph{complete directed graph} $Q_n$ with $n$ vertices is the directed graph with $n$ vertices and all possible $n(n-1)$ arcs.
An undirected graph with $t$ vertices and all possible $t \choose 2$ edges is a \emph{clique} $K_t$.
We say that an undirected graph $G$ \emph{contains} a clique $K_t$, if there is a subset $U\subseteq V(G)$ of $t$ vertices such that $G[U]$ is a clique~$K_t$. 
For a $G$ that does not contain a clique $K_t$, we say that $G$ is \emph{$K_t$-free}.
The \emph{clique number} of an undirected graph~$G$, denoted $\omega(G)$, is the maximum $t$ such that $G$ contains a clique $K_t$.
The natural computational problem is
\defdecproblem{\Pclique}{An undirected graph $G$ and a positive integer $t$}{$\mathsf{Yes}$ if and only if $\omega(G) \ge t$}
This problem has broad practical applications in different areas of discrete optimization and it is also of great theoretical interest. 
It is one of the first problems known to be \NP-complete, see~\cite{GareyJ79}, and is used as a base for countless reductions.

Following the work of Nguyen, Seymour and Scott~\cite{NguyenSS25}, Aboulker, Aubian, Charbit and Lopes~\cite{AboulkerACL23} give the following definition of the clique number for directed graphs.
The \emph{directed clique number} of a directed graph~$G$, denoted $\diomega(G)$, is defined as the minimum number $t$ such that there exists linear order $\ll$ on $V(G)$ with the clique number of the backedge graph $G^\ll$ equal to $t$, i.e.\ 
\[
    \diomega(G) = \min_{\ll} \omega(G^\ll)\text{,}
\]
and the natural computational problem is
\defdecproblem{\Pdirclique}{A directed graph~$G$ and a positive integer~$t$}{$\mathsf{Yes}$ if and only if $\diomega(G) \le t$}
Nguyen, Seymour and Scott~\cite[Section 6]{NguyenSS25} were focused on tournaments and asked the following question.
\textit{,,Can we test whether a tournament has small directed clique number (even in an approximate sense) in polynomial time?
Is it in \CONP?''}
Let us define a variant of the directed clique problem in which the input graph is required to be a tournament.
\defdecproblem{\Ptourclique}{A tournament~$T$ and a positive integer~$t$}{$\mathsf{Yes}$ if and only if $\diomega(T) \le t$}
We also define parametrized variants \Ptdirclique and \Pttourclique in which $t$ is some fixed value instead of being given on the input.

It is an easy observation that $\diomega(G) =0$ iff $G$ is an empty graph, and that $\diomega(G) = 1$ iff $G$ is transitive.
Transitivity can be easily checked using depth first search algorithm, and we get that \Ptdirclique[1] and \Pttourclique[1] are in \POL.
For every $t\ge3$, Aubian~\cite{Aubian24} showed that both \Pttourclique and \Ptdirclique are \NP-complete and conjectured that the same holds for $t=2$.
To our best knowledge, this question has not been resolved yet.

Other graph parameters can also be defined for directed graphs using the backedge graph similarly as for the definition of the directed clique number.
See the discussion of other intersting parameters in a paper on degree-width of directed graphs by Aboulker, Oijid, Petit, Rocton and Simon~\cite{AboulkerOPRS24}.
The classical notion of \emph{dichromatic number}, defined first by Neumann-Lara~\cite{Neumann82} can equivalently be defined for a directed graph $G$ as the minimum chromatic number of $G^\ll$ taken over all linear orders $\ll$ on $V(G)$.
Therefore, out of these similarly defined parameters, the dichromatic number attracts the most attention.
Not surprisingly, the dichromatic number and the directed clique number are related, see~\cite{NguyenSS25,AboulkerACL23,Aubian24,AboulkerOPRS24}.

Nguyen, Seymour and Scott~\cite{NguyenSS25} asked if \Ptourclique is in \POL or \CONP.
We are unable to answer this questions, but we provide some strong evidence that the answer should be negative.
We do that by settling the complexity of the more general problem on directed graphs.
Observe that the problem \Pdirclique on instance $\abrac{G,t}$ is naturally expressed as:
\[
  \exists_{\ll\text{ - a linear order on $V(G)$}}:\quad\forall_{A\subseteq V(G),\norm{A}= t+1}:\quad \text{$A$ does not induce a $(t+1)$-clique in $G^\ll$}\text{,}
\]
and from this formulation we easily get that \Pdirclique is in the second level of polynomial hierarchy class \SPTWO.
Consult textbook by Arora and Barak~\cite[Chapter~5]{AroraB09} for an introduction of the polynomial hierarchy.
Schaefer and Umans~\cite{SchaeferU02_1,SchaeferU02_2,SchaeferU08} give an extensive list of complete problems for different classes in the polynomial hierarchy.
For a very brief introduction, \SPTWO is defined as $\NP^\NP$ -- a class of languages decidable in polynomial time by nondeterministic Turing machines witch access to an \NP-oracle,
where an \NP-oracle allows to test any language in \NP in a single step of execution.
The canonical complete problem for \SPTWO is the following.
\defdecproblem{\PSPTWOSAT}{Boolean formula $\varphi(x_1,\ldots,x_a,y_1,\ldots,y_b)$ with variables in two disjoint sets $\set{x_1,\ldots,x_a}$ and $\set{y_1,\ldots,y_b}$}{$\mathsf{Yes}$ if and only if the following Boolean formula is true.\\
    &$\exists_{x_1,x_2,\ldots,x_a}:\quad \forall_{y_1,y_2,\ldots,y_b}:\quad \varphi(x_1,\ldots,x_a,y_1,\ldots,y_b)$
    }
It was independently proved by Stockmeyer~\cite{Stockmeyer76} and Wrathall~\cite{Wrathall76} that the class \SPTWO is exactly the class of languages reducible to \PSPTWOSAT via polynomial-time many-one reductions.
This fact makes \SPTWO a natural complexity class for \Pdirclique and \Ptourclique problems.
We make the following conjecture that would settle the complexity of both problems.
\begin{conjecture}\label{con:main}
    \Ptourclique is \SPTWO-complete.
\end{conjecture}
An obstacle in proving \cref{con:main} is the lack of constructions of tournaments with high directed clique number.
Aboulker, Aubian, Charbit and Lopes~\cite{AboulkerACL23} give a simple construction of tournaments with $\diomega(T) \ge \Omega(\log \norm{V(T)})$.
It is also true that for a tournament in which direction of every arc is picked uniformly at random, we have $\diomega(T) = \Theta(\log \norm{V(T)})$ asymptotically almost surely.
We are unaware of any construction of tournaments with higher directed clique number.
If \cref{con:main} is true, then we should expect that there are tournaments with directed clique number polynomial in the number of vertices.
To complement the logarithmic lower bounds, we provide a polynomial upper bound on the directed clique number.
In \cref{sec:upper_bound} we give a proof that for any tournament $T$ we have $\diomega(T) \le \sqrt{2\norm{V(T)}}$.
Let us note that for a random order $\ll$ on the vertices of $TT_n$ we expect the clique number of $G^\ll$ to be $\Theta(\sqrt{V(T)})$.
Nevertheless, $\diomega(TT_n)=1$.
We conjecture that the lower bound for the directed clique number of tournaments can be improved significantly.
\begin{conjecture}\label{con:lower_bound}
    There is an explicit construction of tournaments $T_1,T_2,\ldots$ such that $\norm{V(T_n)} = \Theta(n)$ and $\diomega(T_n) = \Omega(n^{1/100})$.
\end{conjecture}

As we are unable to resolve \cref{con:main}, we prove the following weaker result that gives some evidence that \cref{con:main} should be true.
We turn our attention back to \Pdirclique problem, as it is much easier to construct directed graphs with high directed clique number.
In particular, for the complete directed graph $Q_n$ we have $\diomega(Q_n)=n$.
The main result of this paper is the following.
\begin{restatable}{theorem}{thmmain}\label{thm:main}
\Pdirclique is \SPTWO-complete.
\end{restatable}
Our feeling is that the reduction used to prove \cref{thm:main} is robust enough, so that a positive answer to \cref{con:lower_bound} should also give a positive answer to \cref{con:main} using similar techniques.

\section{Main Result}

The proof goes by a reduction from the following problem, which is a variation on the quantified boolean formula satisfaction problem.
It is a natural \SPTWO-complete problem with an easy reduction from \PSPTWOSAT~\cite{Stockmeyer76}.
\defdecproblem{\PSPTWO}{$3$-CNF formula $\varphi(x_1,\ldots,x_a,y_1,\ldots,y_b)$ with variables in two disjoint sets $\set{x_1,\ldots,x_a}$ and $\set{y_1,\ldots,y_b}$}{$\mathsf{Yes}$ if and only if the following Boolean formula is true.\\
 & $\exists_{x_1,x_2,\ldots,x_a}:\quad \neg \exists_{y_1,y_2,\ldots,y_b}:\quad \varphi(x_1,\ldots,x_a,y_1,\ldots,y_b)$}

We are now ready to recall and prove the main theorem.
\thmmain*
\begin{proof}
    We present a reduction from \PSPTWO to \Pdirclique.
    Assume that we are given an instance $\abrac{X,Y,\mathcal{C}}$ of \PSPTWO, where $X=\set{x_1,\ldots,x_a}$, $Y=\set{y_1,\ldots,y_b}$ are two disjoint sets of variables, and $\mathcal{C}=\set{C_1,\ldots,C_c}$ is a set o clauses with each clause having exactly three occurrences of three distinct variables in $X \cup Y$.
    For technical reasons, we assume that $c>6$, as otherwise there are at most $18$ variables and we can simply check all the possibilities.
    We set $t=2c-1$, and construct a graph $G$ such that $G$ admits a linear order~$\ll$ of vertices such that $G^\ll$ does not contain $K_{2c}$ as a subgraph if and only if the following Boolean formula
    \[
        \exists_{x_1,x_2,\ldots,x_a}:\quad \neg \exists_{y_1,y_2,\ldots,y_b}:\quad C_1 \wedge C_2 \wedge \ldots\wedge C_c
    \]
    is true.
    In the next paragraphs we describe gadgets necessary to construct the graph $G$.

    \subparagraph*{Binary Gadget}
    For every occurrence of a variable $x_i \in X$ in a clause $C_k$ we introduce a subgraph $G_{x_i}^{(k)}$, see \cref{fig:binary_gadget}, consisting of the following:
    \begin{enumerate}
        \item Three disjoint complete directed graphs (gray nodes in \cref{fig:binary_gadget}): $A'$ with $t-2$ vertices, and $A_F, A_T$ with $t-1$ vertices each.
        \item Vertices $x_i^{(k)}$, $x_{i,T}^{(k)}$, $x_{i,F}^{(k)}$ and $w_i^{(k)}$.
        \item Arcs as in \cref{fig:binary_gadget}. An arc between a vertex and a gray node represents all possible arcs between the vertex and every vertex in the complete directed graph represented by the gray node.
    \end{enumerate}
\begin{figure}
\begin{center}
\begin{tiny}
\begin{tikzpicture}
\node[circle,draw] at (0,0) (x) {$x_i^{(k)}$};
    \node[circle,draw] at (-2,0) (xF) {$x_{i,F}^{(k)}$};
    \node[circle,draw] at (2,0) (xT) {$x_{i,T}^{(k)}$};
\node[circle,draw,fill=lightgray] at (-2,-4) (xFA) {$A_F=Q_{t-1}$};
\node[circle,draw,fill=lightgray] at (2,-4) (xTA) {$A_T=Q_{t-1}$};
\begin{scope}[on background layer]
    \node (back) [fit=(xF) (xT) (x),fill=blue!25, inner sep=0.5em, rounded corners] {};
    \node (title) at (back.north) [anchor=south,color=blue] {\small $G_{x_i}^{(k)}$};
    \node [fit=(title) (xF) (xT) (x) (xFA) (xTA),fill=blue!5, inner sep=1em, rounded corners] {};
    \node (back) [fit=(xF) (xT) (x),fill=blue!25, inner sep=0.5em, rounded corners] {};
    \node (title) at (back.north) [anchor=south,color=blue] {\small $G_{x_i}^{(k)}$};
\end{scope}
\path[draw,->,red,very thick] (x) -- (xF);
\path[draw,->,red,very thick] (xT) -- (x);
    \node[circle,draw] at (0,-4) (xA) {$w_{i}^{(k)}$};
\path[draw,->] (xF) -- (xA);
\path[draw,->] (xA) -- (xT);
\path[draw,->] (xF) edge [bend left=20] (xT);
\path[draw,->] (xT) edge [bend left=20] (xF);
\node[circle,draw,fill=lightgray] at (-2,-4) (xFA) {$A_F=Q_{t-1}$};
\path[draw,->] (xFA) edge [bend left=10] (xA);
\path[draw,->] (xA) edge [bend left=10] (xFA);
\path[draw,->] (xFA) edge [bend left=10] (xF);
\path[draw,->] (xF) edge [bend left=10] (xFA);
\node[circle,draw,fill=lightgray] at (2,-4) (xTA) {$A_T=Q_{t-1}$};
\path[draw,->] (xTA) edge [bend left=10] (xA);
\path[draw,->] (xA) edge [bend left=10] (xTA);
\path[draw,->] (xTA) edge [bend left=10] (xT);
\path[draw,->] (xT) edge [bend left=10] (xTA);
\node[circle,draw,fill=lightgray] at (0,-2) (xM) {$A'=Q_{t-2}$};
\path[draw,->] (xM) edge [bend left=10] (xT);
\path[draw,->] (xT) edge [bend left=10] (xM);
\path[draw,->] (xM) edge [bend left=10] (xF);
\path[draw,->] (xF) edge [bend left=10] (xM);
\path[draw,->] (xM) edge [bend left=10] (x);
\path[draw,->] (x) edge [bend left=10] (xM);
\end{tikzpicture}
\end{tiny}
\end{center}
    \caption{\label{fig:binary_gadget} Binary gadget for occurrence of variable $x_i$ in clause $C_k$.
}
\end{figure}
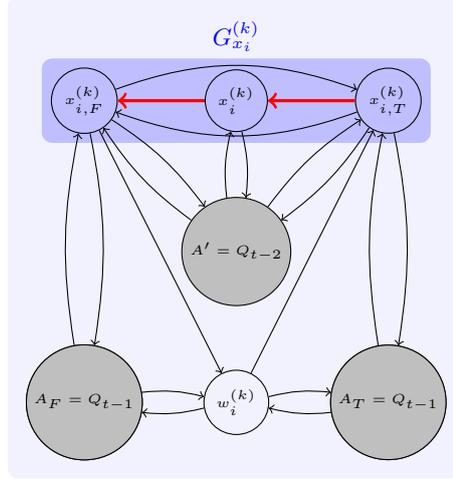
    Vertices $x_i^{(k)}$, $x_{i,T}^{(k)}$, and $x_{i,F}^{(k)}$ (in dark blue box in \cref{fig:binary_gadget}) are later connected to other vertices in the graph.
    All the other vertices in the gadget, i.e.\ vertices in $A_F\cup A_T \cup A' \cup \set{w_i^{(k)}}$ are \emph{internal} for the gadget and do not get any more adjacencies in the graph.
    The following claim captures the crucial property of the gadget.
    \begin{claim}\label{clm:binary_gadget}
      For any linear order $\ll$ on the vertices of $G_{x_i}^{(k)}$ such that the backedge graph $G_{x_i}^{(k)\ll}$ is $K_{2c}$-free we have
    that $\backedge(G_{x_i}^{(k)},\ll)$ contains exactly one of the edges $\{x_i^{(k)},x_{i,F}^{(k)}\}$ or $\{x_i^{(k)},x_{i,T}^{(k)}\}$ (arising from red arcs in \cref{fig:binary_gadget}).
\end{claim}
    \begin{claimproof}
        Let $\ll$ be any linear order on the vertices of $G_{x_i}^{(k)}$.
        Assuming the backedge graph is $K_{2c}$-free, we get $x_{i,F}^{(k)} \ll w_i^{(k)}$, as otherwise the $t-1$ vertices in $A_F$, $x_{i,F}^{(k)}$ and $w_{i}^{(k)}$ span a $2c$-clique in the backedge graph.
        Similarly, using the $t-1$ vertices in $A_T$, $x_{i,T}^{(k)}$, and $w_i^{(k)}$, we get $w_i^{(k)} \ll x_{i,T}^{(k)}$, and as a consequence $x_{i,F}^{(k)} \ll x_{i,T}^{(k)}$.
        If $x_{i,F}^{(k)} \ll x_i^{(k)} \ll x_{i,T}^{(k)}$ then these three vertices and the $t-2$ vertices in $A'$ span a $2c$-clique in the backedge graph.
        Thus, exactly one of $x_i^{(k)} \ll x_{i,F}^{(k)}$ or $x_{i,T}^{(k)} \ll x_i^{(k)}$ holds and the claim follows.
    \end{claimproof}

    The idea behind the gadget is that the choice of the red edge remaining in the backedge graph corresponds to the valuation of variable $x_i$ in clause $C_k$.

    \subparagraph*{Copy Gadget}
    For every two occurrences of a single variable $x_i \in X$ in two different clauses $C_k$ and $C_\ell$ we introduce a subgraph $G_{x_i}^{(k,\ell)}$, see \cref{fig:copy_gadget}, consisting of the following:
    \begin{enumerate}
        \item Additional copy of a binary gadget (the middle layer in \cref{fig:copy_gadget}).
        \item Four disjoint complete directed graphs (gray nodes in \cref{fig:binary_gadget}) $A_1,A_2,A_3,A_4$ with $t-3$ vertices each.
        \item Arcs as in \cref{fig:copy_gadget}.
    \end{enumerate}
\begin{figure}
\begin{center}
\begin{tiny}
\begin{tikzpicture}[scale=1.5]
    \node[circle,draw] at (0,0) (x) {$x_i^{(k)}$};
    \node[circle,draw] at (-3,0) (xF) {$x_{i,F}^{(k)}$};
    \node[circle,draw] at (3,0) (xT) {$x_{i,T}^{(k)}$};
\begin{scope}[on background layer]
    \node (back) [fit=(xF) (xT) (x),fill=blue!25, inner sep=0.5em, rounded corners] {};
    \node (title) at (back.north) [anchor=south,color=blue] {\small $G_{x_i}^{(k)}$};
    \node [fit=(title) (xF) (xT) (x),fill=blue!5, inner sep=1em, rounded corners] {};
    \node (back) [fit=(xF) (xT) (x),fill=blue!25, inner sep=0.5em, rounded corners] {};
    \node (title) at (back.north) [anchor=south,color=blue] {\small $G_{x_i}^{(k)}$};
\end{scope}
\path[draw,->,red, very thick] (x) -- (xF);
\path[draw,->,red, very thick] (xT) -- (x);
\path[draw,<->] (xF) edge [bend left=10] (xT);
    \node[circle,draw] at (0,-3) (y) {$x_{i}^{(k,\ell)}$};
    \node[circle,draw] at (3,-3) (yF) {$x_{i,F}^{(k,\ell)}$};
    \node[circle,draw] at (-3,-3) (yT) {$x_{i,T}^{(k,\ell)}$};
\begin{scope}[on background layer]
    \node (back2) [fit=(yF) (yT) (y),fill=blue!5, inner sep=1em, rounded corners] {};
    \node (back2) [fit=(yF) (yT) (y),fill=blue!25, inner sep=0.5em, rounded corners] {};
\end{scope}
\path[draw,->,red, very thick] (y) -- (yF);
\path[draw,->,red, very thick] (yT) -- (y);
\path[draw,->] (yF) edge [bend left=15] (yT);
\path[draw,->] (yT) edge [bend left=15] (yF);
\node[circle,draw,fill=lightgray] at (-1.5,-1.5) (xA) {$A_1=Q_{t-3}$};
\path[draw,->] (xA) edge [bend left=10] (x);
\path[draw,->] (x) edge [bend left=10] (xA);
\path[draw,->] (xA) edge [bend left=10] (y);
\path[draw,->] (y) edge [bend left=10] (xA);
\path[draw,->] (xA) edge [bend left=10] (xF);
\path[draw,->] (xF) edge [bend left=10] (xA);
\path[draw,->] (xA) edge [bend left=10] (yT);
\path[draw,->] (yT) edge [bend left=10] (xA);
\path[draw,->] (yF) edge [bend left=25] (x);
\path[draw,->] (x) edge [bend left=25] (yF);
\path[draw,->] (xT) edge [bend left=25] (y);
\path[draw,->] (y) edge [bend left=25] (xT);
\path[draw,->] (xF) edge [bend left=25] (y);
\path[draw,->] (y) edge [bend left=25] (xF);
\path[draw,->] (yT) edge [bend left=25] (x);
\path[draw,->] (x) edge [bend left=25] (yT);

\node[circle,draw,fill=lightgray] at (1.5,-1.5) (xB) {$A_2=Q_{t-3}$};
\path[draw,->] (xB) edge [bend left=10] (x);
\path[draw,->] (x) edge [bend left=10] (xB);
\path[draw,->] (xB) edge [bend left=10] (y);
\path[draw,->] (y) edge [bend left=10] (xB);
\path[draw,->] (xB) edge [bend left=10] (xT);
\path[draw,->] (xT) edge [bend left=10] (xB);
\path[draw,->] (xB) edge [bend left=10] (yF);
\path[draw,->] (yF) edge [bend left=10] (xB);

\path[draw,->] (x) edge [bend left=10] (y);
\path[draw,->] (y) edge [bend left=10] (x);
\path[draw,->] (xF) edge [bend left=10] (yT);
\path[draw,->] (yT) edge [bend left=10] (xF);
\path[draw,->] (xT) edge [bend left=10] (yF);
\path[draw,->] (yF) edge [bend left=10] (xT);

    \node[circle,draw] at (0,-6) (x2) {$x_i^{(\ell)}$};
    \node[circle,draw] at (-3,-6) (x2F) {$x_{i,F}^{(\ell)}$};
    \node[circle,draw] at (3,-6) (x2T) {$x_{i,T}^{(\ell)}$};
\begin{scope}[on background layer]
    \node (back3) [fit=(x2F) (x2T) (x2),fill=blue!25, inner sep=0.5em, rounded corners] {};
    \node (title3) at (back3.south) [anchor=north,color=blue] {\small $G_{x_i}^{(\ell)}$};
    \node [fit=(title3) (x2F) (x2T) (x2),fill=blue!5, inner sep=1em, rounded corners] {};
    \node (back3) [fit=(x2F) (x2T) (x2),fill=blue!25, inner sep=0.5em, rounded corners] {};
    \node (title3) at (back3.south) [anchor=north,color=blue] {\small $G_{x_i}^{(\ell)}$};
\end{scope}
\path[draw,->,red, very thick] (x2) -- (x2F);
\path[draw,->,red, very thick] (x2T) -- (x2);
\path[draw,<->] (x2F) edge [bend right=10] (x2T);
\node[circle,draw,fill=lightgray] at (-1.5,-4.5) (x2A) {$A_3=Q_{t-3}$};
\path[draw,->] (x2A) edge [bend left=10] (x2);
\path[draw,->] (x2) edge [bend left=10] (x2A);
\path[draw,->] (x2A) edge [bend left=10] (y);
\path[draw,->] (y) edge [bend left=10] (x2A);
\path[draw,->] (x2A) edge [bend left=10] (x2F);
\path[draw,->] (x2F) edge [bend left=10] (x2A);
\path[draw,->] (x2A) edge [bend left=10] (yT);
\path[draw,->] (yT) edge [bend left=10] (x2A);

\node[circle,draw,fill=lightgray] at (1.5,-4.5) (x2B) {$A_4=Q_{t-3}$};
\path[draw,->] (x2B) edge [bend left=10] (x2);
\path[draw,->] (x2) edge [bend left=10] (x2B);
\path[draw,->] (x2B) edge [bend left=10] (y);
\path[draw,->] (y) edge [bend left=10] (x2B);
\path[draw,->] (x2B) edge [bend left=10] (x2T);
\path[draw,->] (x2T) edge [bend left=10] (x2B);
\path[draw,->] (x2B) edge [bend left=10] (yF);
\path[draw,->] (yF) edge [bend left=10] (x2B);

\path[draw,->] (x2) edge [bend left=10] (y);
\path[draw,->] (y) edge [bend left=10] (x2);
\path[draw,->] (x2F) edge [bend left=10] (yT);
\path[draw,->] (yT) edge [bend left=10] (x2F);
\path[draw,->] (x2T) edge [bend left=10] (yF);
\path[draw,->] (yF) edge [bend left=10] (x2T);

\path[draw,->] (yF) edge [bend left=25] (x2);
\path[draw,->] (x2) edge [bend left=25] (yF);
\path[draw,->] (x2T) edge [bend left=25] (y);
\path[draw,->] (y) edge [bend left=25] (x2T);
\path[draw,->] (x2F) edge [bend left=25] (y);
\path[draw,->] (y) edge [bend left=25] (x2F);
\path[draw,->] (yT) edge [bend left=25] (x2);
\path[draw,->] (x2) edge [bend left=25] (yT);
\end{tikzpicture}
\end{tiny}
\end{center}
\caption{\label{fig:copy_gadget} Copy gadget for occurrences of variable $x_i$ in clauses $C_k$ and $C_\ell$.
    Blue boxes are copies of the binary gadget.
    }
\end{figure}
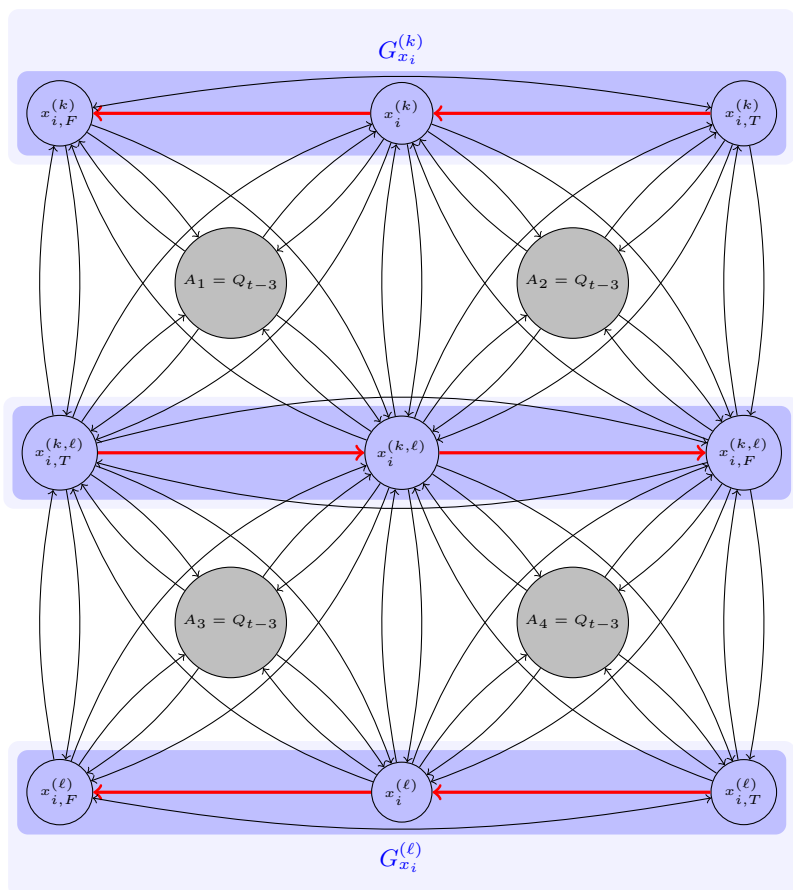
    All the new vertices vertices introduced in the gadget are internal for the gadget and do not get any more adjacencies in the graph.
    Note that the gadget does not introduce any arcs between the vertices of $G_{x_i}^{(k)}$ and the vertices of $G_{x_i}^{(\ell)}$. 
    The following claim captures the crucial property of the gadget.
    \begin{claim}\label{clm:copy_gadget}
      For any linear order $\ll$ on the vertices of $G_{x_i}^{(k)}$, $G_{x_i}^{(\ell)}$, and $G_{x_i}^{(k,\ell)}$ such that the backedge graph $G_{x_i}^{(k)\ll}$ is $K_{2c}$-free we have
        that $\backedge(G_{x_i}^{(k,\ell)},\ll)$ contains the edge $\{x_i^{(k)},x_{i,F}^{(k)}\}$ if and only if it contains the edge $\{x_i^{(\ell)},x_{i,F}^{(\ell)}\}$.
\end{claim}
    \begin{claimproof}
        Assume that $\{x_i^{(k)},x_{i,F}^{(k)}\}$ is an edge in $\backedge(G_{x_i}^{(k,\ell)},\ll)$.
        We get that $\{x_i^{(k,\ell)},x_{i,T}^{(k,\ell)}\}$ is not an edge in $\backedge(G_{x_i}^{(k,\ell)},\ll)$,
        as otherwise the $t-3$ vertices in $A_1$ and the four surrounding vertices induce a $2c$-clique in the backedge graph.
        As the middle layer is a copy of the binary gadget, by \cref{clm:binary_gadget}, we get that $\{x_i^{(k,\ell)},x_{i,F}^{(k,\ell)}\}$ is an edge in $\backedge(G_{x_i}^{(k,\ell)},\ll)$.
        We get that $\{x_i^{(\ell)},x_{i,T}^{(\ell)}\}$ is not an edge in $\backedge(G_{x_i}^{(k,\ell)},\ll)$,
        as otherwise the $t-3$ vertices in $A_4$ and the four surrounding vertices induce a $2c$-clique in the backedge graph.
        By \cref{clm:binary_gadget} applied to the binary gadget $G_{x_i}^{(\ell)}$, we get that $\{x_i^{(\ell)},x_{i,F}^{(\ell)}\}$ is an edge in $\backedge(G_{x_i}^{(k,\ell)},\ll)$.

        Similarly, if $\{x_i^{(k)},x_{i,T}^{(k)}\}$ is an edge in $\backedge(G_{x_i}^{(k,\ell)},\ll)$, we get that $\{x_i^{(\ell)},x_{i,T}^{(\ell)}\}$ is an edge in $\backedge(G_{x_i}^{(k,\ell)},\ll)$,
        and the claim follows.
   \end{claimproof}
    
    The idea behind the gadget is to ensure that red edges corresponding to all occurrences of a single variable in different clauses give a consistent valuation of variable $x_i$.

    \subparagraph*{Clause gadget}
    For every occurrence of a variable $y_j \in Y$ in $C_k$, we introduce two vertices $y_{j,A}^{(k)}$, $y_{j,B}^{(k)}$ connected by two parallel arcs (in opposite directions).
    For every clause \(C_k\) we use 6 vertices in \(G\) in three groups: a pair of vertices for every occurrence \(z\) in \(C_k\).
    \begin{itemize}
      \item If \(z = y_j\) or \(z = \neg y_j\) for some \(y_j \in Y\), then the group consists of vertices $y_{j,A}^{(k)}$ and $y_{j,B}^{(k)}$.
        \item If \(z = x_i\) for some \(x_i \in X\), then the group consists of vertices $x_{i}^{(k)}$ and $x_{i,T}^{(k)}$ introduced in $G_{x_i}^{(k)}$.
        \item If \(z = \neg x_i\) for some \(x_i \in X\), then the group consists of vertices $x_{i}^{(k)}$ and $x_{i,F}^{(k)}$ introduced in $G_{x_i}^{(k)}$.
    \end{itemize}
\begin{figure}
\begin{center}
\begin{tiny}
\begin{tikzpicture}
    \node[circle,draw] at (0,0) (x11) {$x_1^{(1)}$};
    \node[circle,draw] at (0,-1.5) (x11T) {$x_{1,T}^{(1)}$};
    \path[draw,->,thick] (x11T) -- (x11);
\begin{scope}[on background layer]
    \node (back11) [fit=(x11) (x11T),fill=black!10, inner sep=0.5em, rounded corners] {};
    \node at (back11.west) [anchor=east,color=black] {$x_1$};
\end{scope}
    \node[circle,draw] at (0,-3) (x21) {$x_2^{(1)}$};
    \node[circle,draw] at (0,-4.5) (x21F) {$x_{2,F}^{(1)}$};
    \path[draw,->,thick] (x21) -- (x21F);
\begin{scope}[on background layer]
    \node (back21) [fit=(x21) (x21F),fill=black!10, inner sep=0.5em, rounded corners] {};
    \node at (back21.west) [anchor=east,color=black] {$\neg x_2$};
\end{scope}
    \node[circle,draw] at (0,-6) (y31) {$y_{3,A}^{(1)}$};
    \node[circle,draw] at (0,-7.5) (y31C) {$y_{3,B}^{(1)}$};
    \path[draw,->] (y31) edge [bend left=10] (y31C);
    \path[draw,->] (y31C) edge [bend left=10] (y31);
\begin{scope}[on background layer]
    \node (back31) [fit=(y31) (y31C),fill=black!10, inner sep=0.5em, rounded corners] {};
    \node at (back31.west) [anchor=east,color=black] {$y_3$};
\end{scope}

    \node[circle,draw] at (6,0) (x12) {$x_1^{(2)}$};
    \node[circle,draw] at (6,-1.5) (x12F) {$x_{1,F}^{(2)}$};
    \path[draw,->,thick] (x12) -- (x12F);
\begin{scope}[on background layer]
    \node (back12) [fit=(x12) (x12F),fill=black!10, inner sep=0.5em, rounded corners] {};
    \node at (back12.east) [anchor=west,color=black] {$\neg x_1$};
\end{scope}
    \node[circle,draw] at (6,-3) (y32) {$y_{3,A}^{(2)}$};
    \node[circle,draw] at (6,-4.5) (y32C) {$y_{3,B}^{(2)}$};
    \path[draw,->] (y32) edge [bend left=10] (y32C);
    \path[draw,->] (y32C) edge [bend left=10] (y32);
\begin{scope}[on background layer]
    \node (back32) [fit=(y32) (y32C),fill=black!10, inner sep=0.5em, rounded corners] {};
    \node at (back32.east) [anchor=west,color=black] {$\neg y_3$};
\end{scope}
    \node[circle,draw] at (6,-6) (y42) {$y_{4,A}^{(2)}$};
    \node[circle,draw] at (6,-7.5) (y42C) {$y_{4,B}^{(2)}$};
    \path[draw,->] (y42) edge [bend left=10] (y42C);
    \path[draw,->] (y42C) edge [bend left=10] (y42);
\begin{scope}[on background layer]
    \node (back42) [fit=(y42) (y42C),fill=black!10, inner sep=0.5em, rounded corners] {};
    \node at (back42.east) [anchor=west,color=black] {$y_4$};
\end{scope}

    \coordinate (a1t) at (back11.north east);
    \coordinate (a1b) at (back11.south east);
    \coordinate (a2t) at (back21.north east);
    \coordinate (a2b) at (back21.south east);
    \coordinate (a3t) at (back31.north east);
    \coordinate (a3b) at (back31.south east);

    \coordinate (b1t) at (back12.north west);
    \coordinate (b1b) at (back12.south west);
    \coordinate (b2t) at (back32.north west);
    \coordinate (b2b) at (back32.south west);
    \coordinate (b3t) at (back42.north west);
    \coordinate (b3b) at (back42.south west);

    %\path[draw,<->,thick,double] (barycentric cs:a1t=3,a1b=2) -- (barycentric cs:b1t=3,b1b=2);
    \path[draw,<->,thick,double] (barycentric cs:a1t=2,a1b=2) -- (barycentric cs:b2t=3,b2b=2);
    \path[draw,<->,thick,double] (barycentric cs:a1t=2,a1b=3) -- (barycentric cs:b3t=3,b3b=2);
    \path[draw,<->,thick,double] (barycentric cs:a2t=3,a2b=2) -- (barycentric cs:b1t=2,b1b=2);
    \path[draw,<->,thick,double] (barycentric cs:a2t=2,a2b=2) -- (barycentric cs:b2t=2,b2b=2);
    \path[draw,<->,thick,double] (barycentric cs:a2t=2,a2b=3) -- (barycentric cs:b3t=2,b3b=2);
    \path[draw,<->,thick,double] (barycentric cs:a3t=3,a3b=2) -- (barycentric cs:b1t=2,b1b=3);
    %\path[draw,<->,thick,double] (barycentric cs:a3t=2,a3b=2) -- (barycentric cs:b2t=2,b2b=3);
    \path[draw,<->,thick,double] (barycentric cs:a3t=2,a3b=3) -- (barycentric cs:b3t=2,b3b=3);
\end{tikzpicture}
\end{tiny}
\end{center}
    \caption{\label{fig:clause_gadget} Clause gadgets for clauses $C_1 = x_1 \vee \neg x_2 \vee y_3$ (left) and $C_2 = \neg x_1 \vee \neg y_3 \vee y_4$ (right).
    Double edges represent all possible edges going from one group to another in any direction.
    There are no edges between vertices corresponding to $x_1$ in $C_1$ and in $C_2$, as these are negated occurrences.
    Similarly for occurrences of $y_3$.
    }
\end{figure}
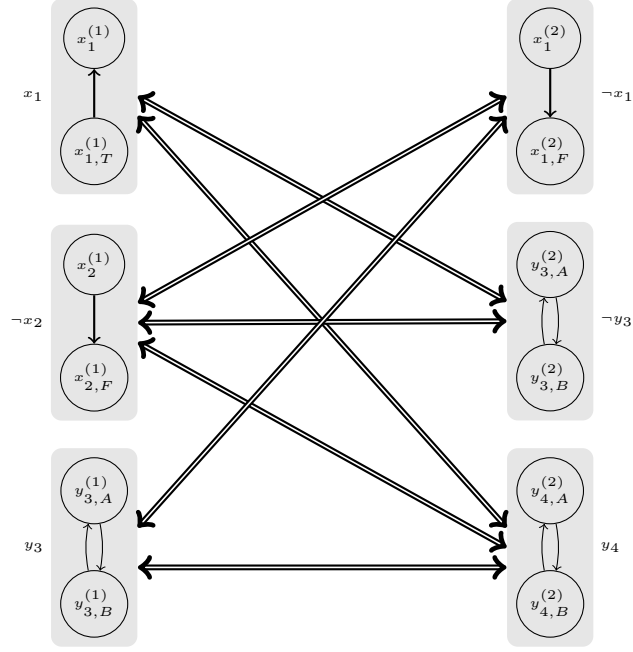

    \subparagraph*{Reduction}
    We construct graph $G$ by introducing a binary gadget for every occurrence of a variable in $X$, a copy gadget for every pair of occurrences of a single variable in $X$, and a clause gadget for every clauses in $\mathcal{C}$.
    Finally, for any two groups in two different clause gadgets if the groups correspond to different variables or correspond to the occurrences of the same variable with the same sign (both occurrences are positive or both are negative) we add all possible eight arcs between the two groups, see \cref{fig:clause_gadget} for an example.
    This concludes the description of the reduction.
    Note that 
    the size of the constructed graph is polynomial in the size of the input formula
    and the reduction can be performed in polynomial time. 

    \subparagraph*{Completeness}
    Assuming that $\abrac{X,Y,\mathcal{C}}$ is a $\mathsf{Yes}$ instance, 
    let $\nu=(\nu_1,\nu_2,\ldots,\nu_a)$ be a valuation of $x_1,x_2,\ldots,x_a$ such that the formula $$\exists_{y_1,y_2,\ldots,y_b}:\quad (C_1 \wedge C_2 \wedge \ldots \wedge C_c)[x_1=\nu_1,x_2=\nu_2,\ldots,x_a=\nu_a]$$ is false.
    We construct a linear order \(\ll\) such that the backedge graph \(G^{\ll}\) does not contain a clique of size \(2c\).

    First, we use valuation $\nu$ to define the order in each binary gadget.
    Let $G_{x_i}^{(k)}$ be some binary gadget.
    If $\nu_i$ is true, we set $x_i^{(k)} \ll x_{i,F}^{(k)} \ll w_i^{(k)} \ll x_{i,T}^{(k)}$.
    Otherwise, we set $x_{i,F}^{(k)} \ll w_i^{(k)} \ll x_{i,T}^{(k)} \ll x_i^{(k)}$.
    We extend this order to the remaining internal vertices of the gadget arbitrarily, as the resulting backedge graph is the same for any resulting order.
    We show that the backedge graph induced by the vertices of the binary gadget is $K_{2c}$-free.
    As we have assumed $c>6$, any clique with $2c$ vertices in the backedge graph must use at least one vertex in $A_F$, $A_T$ or $A'$.
    As there are no arcs between these three sets, we can assume that the clique consists of all vertices in exactly one of the sets $A_F$, $A_t$ or $A'$ and possibly some of the remaing four vertices.
    Now, a simple case analysis similar to the proof of \cref{clm:binary_gadget} shows that the backedge graph induced by the vertices of the binary gadget is $K_{2c}$-free.
    If $\nu_i$ is true, the edge $\{x_i^{(k)}, x_{i,T}^{(k)}\}$ is in the resulting backedge graph.
    Otherwise, the edge $\{x_i^{(k)}, x_{i,F}^{(k)}\}$ is in the backedge graph.
    This way, the choice of the red edge in the backedge graph corresponds to the valuation $\nu_i$ of variable $x_i$.
    Observe that as the internal vertices of the gadget are not connected to any other vertices in the graph, no matter how we set the order $\ll$ on the remaining vertices, none of the internal vertices can be an element of $K_{2c}$ in the backedge graph of the whole graph.

    Similarly, we use valuation $\nu$ to define the order in each copy gadget.
    Let $G_{x_i}^{(k,\ell)}$ be some copy gadget.
    If $\nu_i$ is true, we set $x_i^{(k,\ell)} \ll x_{i,F}^{(k,\ell)} \ll w_i^{(k,\ell)} \ll x_{i,T}^{(k,\ell)}$.
    Otherwise, we set $x_{i,F}^{(k,\ell)} \ll w_i^{(k,\ell)} \ll x_{i,T}^{(k,\ell)} \ll x_i^{(k,\ell)}$.
    We extend this order to the remaining internal vertices of the gadget arbitrarily, as the resulting backedge graph is the same for any resulting order.
    Observe that when occurrence of $x_i$ is of the same sign in clause $C_k$ and clause $C_\ell$, then the reduction adds eight additional arcs between vertices of $G_{x_i}^{(k)}$ and vertices of $G_{x_i}^{(\ell)}$ that are not present in \cref{fig:copy_gadget}.
    Even with these additional edges, we show that the backedge graph induced by the vertices of the binary gadget is $K_{2c}$-free.
    As we have assumed $c>6$, any clique with $2c$ vertices in the backedge graph must use at least one vertex in $A_1$, $A_2$, $A_3$ or $A_4$.
    As there are no arcs between these four sets, we can assume that the clique consists of all vertices in exactly one of the sets $A_1$, $A_2$, $A_3$ or $A_4$ and possibly some of the remaing vertices.
    Now, a simple case analysis similar to the proof of \cref{clm:copy_gadget} shows that the backedge graph induced by the vertices of the copy gadget is $K_{2c}$-free.
    As in a binary gadget, none of the internal vertices can be an element of $K_{2c}$ in the backedge graph of the whole graph.

    We combine the orders defined for all binary and all copy gadgets arbitrarily to get a linear order on all vertices introduced in these gadgets.
    Finally, we insert the vertices introduced in clause gadgets for occurrences of variables in $Y$ into order $\ll$ arbitrarily.
    This concludes the construction of the linear order $\ll$.

    We now have to show that the backedge graph $G^{\ll}$ does not contain a clique of size $2c$.
    Assuming to the contrary let $X$ be a set of vertices that induce a clique of size $2c$ in the backedge graph.
    As observed before, none of the internal vertices is in $X$.
    There are at most two vertices in $X$ in each binary gadget, as one of the three non-internal vertices is not connected to any other vertex outside the gadget.
    Similarly, there are at most two vertices in $X$ in each clause gadget, as the vertices in different groups in a single clause gadget are not connected.
    As $\norm{X} = 2c$, we get that there are exactly two vertices in $X$ in each clause gadget and the two vertices in a single clause gadget are in a single group that corresponds to some occurrence $z$ in the clause $C_k$.
    If $z$ is an occurrence of some $x_i \in X$ then we must have that the corresponding red edge is left in the backedge graph and we get that $\nu_i$ satisfies clause $C_k$.
    Now, we construct a valuation $\mu=(\mu_1,\mu_2,\ldots,\mu_b)$ of variables $y_1,y_2,\ldots,y_b$ as follows.
    If $z$ is an occurrence of some $y_j \in Y$ then we set $\mu_j$ so that it satisfies clause $C_k$.
    As there are no edges between groups corresponding to opposite occurrences of a single variable $y_j$ in different clause gadgets, this valuation of variables in $Y$ is consistent.
    If valuation $\mu$ for some variable in $Y$ is not decided yet, we set it arbitrarily.
    We get that the resulting valuation $\mu$ together with $\nu$ satisfies every clause in $C_k$.
    A contradiction.

    \subparagraph*{Soundness}
    Assuming that $\abrac{G,t}$ is a $\mathsf{Yes}$ instance,
    let $\ll$ be a linear order on vertices of $G$ such that in \(G^{\ll}\) there is no clique of size \(2c\).
    \cref{clm:binary_gadget} and \cref{clm:copy_gadget} together imply that we can construct a valuation $\nu$ of the variables in $X$ in the following way.
    We set $\nu_i$ to true if the edge $\{x_i^{(k)}, x_{i,T}^{(k)}\}$ is in the backedge graph for some occurrence of $x_i$ in some clause $C_k$.
    By \cref{clm:copy_gadget}, this means that edge $\{x_i^{(\ell)}, x_{i,T}^{(\ell)}\}$ is in the backedge graph for every occurrence of $x_i$ in any clause $C_\ell$.
    Otherwise, we set $\nu_i$ to false.
    
    We now have to show that the formula $$\exists_{y_1,y_2,\ldots,y_b}:\quad (C_1 \wedge C_2 \wedge \ldots \wedge C_c)[x_1=\nu_1,x_2=\nu_2,\ldots,x_a=\nu_a]$$ is false.
    Assume to the contrary that there is a valuation $\mu=(\mu_1,\mu_2,\ldots,\mu_b)$ of $y_1,y_2,\ldots,y_b$ that together with $\nu$ satisfies every clause in the formula.
    We will use this valuation to construct a clique of size \(2c\) in the backedge graph \(G^{\ll}\).
    For every clause \(C_k\) we pick two vertices that form a group that corresponds to some variable that satisfies $C_k$.
    The construction of the reduction guarantees that there are all four possible edges between any two such groups in different clause gadgets.
    If the selected group corresponds to a variable in $Y$, then the vertices in the group are connected by an edge in the backedge graph.
    If the selected group corresponds to a variable in $X$, then by the definition of valuation $\nu$ the corresponding red edge is in the backedge graph.
    Thus, the selected vertices induce a clique of size \(2c\) in the backedge graph.
    A contradiction.
\end{proof}

\section{Upper bound}\label{sec:upper_bound}
In this section we provide a simple upper bound on the directed clique number of a tournament.

\begin{theorem}\label{thm:upper_bound}
    For every tournament \(T\) with \(\diomega(T) \ge t \ge 1\), we have \(\norm{V(T)} \geq {{t+1}\choose{2}}\).
\end{theorem}

\begin{proof}
    The proof goes by the induction on $t$.
    For the base case $t=1$, we obviously have that the tournament needs to have at least $1={{1+1} \choose 2}$ vertex.
    Now, for the induction step, let $t\ge2$, and assume that the statement of the theorem holds for $t-1$.

    Let $T$ be a tournament with $\diomega(T) \ge t$.
    As $\diomega(T) \ge t$, then for every linear order of $V(T)$, the backedge graph contains a clique $K_t$.
    Thus, there is at least one subset $U \subseteq V(G)$ of vertices that induces a copy of a transitive tournament $TT_t$.
    Let $U' = V(G)\setminus U$ be the complement of $U$, $T' = T[U']$ be the tournament induced by $U'$, $t' = \diomega(T')$ be the directed clique number of $T'$, and $\ll'$ be some linear order of $U'$ with $\omega(T'^{\ll'})=t'$.

    As $U$ induces a transitive subtournament in $T$, we can define order $\ll''$ on $U$ such that for any $x,y \in U$ we have $x \ll'' y$ if and only if $(x,y) \in A(T)$.
    Consider a linear order $\ll$ on $V(G)$ that first puts all vertices in $U'$ in order $\ll'$ and then puts all the vertices in $U$ in order $\ll''$.
    As $\diomega(T) \ge t$, there is at least one copy of a clique $K_t$ in the backedge graph $T^\ll$.
    We have that there are no edges between vertices in $U$ in the backedge graph $T^\ll$.
    Thus, any clique in the backedge graph $T^\ll$ can include at most one vertex in $U$, and we have $t' \ge t-1$.
    By the induction hypothesis, we get $\norm{U'} \ge {t\choose2}$, and
    \[
        \norm{V(G)} = \norm{U'} + \norm{U} \ge {t\choose2} + t = {{t+1}\choose2}\text{.}
    \]
\end{proof}

\begin{corollary}\label{cor:upper_bound}
    For every tournament $T$, we have $\diomega(T) \le \sqrt{2\norm{V(T)}}$.
\end{corollary}
\begin{proof}
    Assuming to the contrary \(\diomega(T) > \sqrt{2\norm{V(T)}}\), we have \({{\diomega(T)+1}\choose{2}} > \norm{V(T)}\), and a contradiction with \cref{thm:upper_bound}.
\end{proof}
    
\bibliography{directed_clique}

\end{document}